\documentclass[12pt]{article}
\usepackage{amssymb,amsmath,amsthm}
\usepackage[mathscr]{eucal}
\usepackage[english]{babel}
\usepackage[utf8]{inputenc}
\usepackage{graphicx}
\newtheorem{theorem}{Theorem}

\newtheorem{defi}{Definition}

\newtheorem{lemma}{Lemma}

\begin{document}

\title{
A Metric Graph for Which the Number of Possible Endpoints of a Random Walk Grows Minimally}
\author{ V.\,L. Chernyshev\footnote{National Research University Higher School of Economics, vchern@gmail.com}, A.\,A. Tolchennikov\footnote{Lomonosov Moscow State University, Institute for Problems in Mechanics (IPMech RAS) tolchennikovaa@gmail.com}}

\maketitle

\abstract{We prove that metric graph with the minimal growth of the number of possible endpoints of a random walk is the union of several linear paths coming out of the same vertex}

\section{Introduction. Formulation of the problem.}

Let $\Gamma = (V,E)$ be a metric graph, where $V$ is the set of vertices and $E$ is the set of edges. We assume that the edges $e_1,\ldots,e_{|E|}$ have lengths $t_1,\ldots,t_{|E|}$ respectively (we will also write $t_e$ for $e\in E$), and these numbers are linearly independent over $\mathbb{Q}$, which corresponds to a situation of general position.
Let one of the vertices be selected in the metric graph. We will call it a \textit{source}. The random walk starts at the source. At each vertex of the metric graph, the next edge is selected with nonzero probability. Accordingly, reflection occurs in hanging vertices. Edge turns are not allowed. The endpoint (final position) of such a walk can be not only a vertex, but also any point on an arbitrary edge of the metric graph.

From counting the number of possible endpoints of a random walk on edges, it is convenient to pass to the following dynamical system (see \cite{RW}), the study of which is motivated by the problem of studying the behavior of wave packets initially localized in a small neighborhood of a single point and propagating on metric graphs or hybrid spaces (see \cite{ETDS} and links therein).
At the initial moment of time, points move from the source vertex along all incident edges with unit speed. At the time when $k$ points, where $k$ can take values from $1$ to valency $\rho$, come to a vertex, $\rho$ points appear, which go out along all edges incident to vertex.
Let us denote by $N(T)$ the number of points that move along the graph by time $T$. The function $N(T)$ is piecewise constant, nondecreasing.
We assume that the value of $N(T)$ at the discontinuity points is equal to half the sum of the limits on the left and on the right. It turns out that the function $N(T)$ has a polynomial asymptotics as $T\to \infty$ (see paper \cite{RCD}, were the polynomial approximating $N(T)$ is described, and references therein).

\par

Let us take the number of edges $|E|$ and the incommensurable lengths of the edges $t_1,\ldots, t_{|E|}$. Let us consider all possible connected graphs $\Gamma$ built from these edges.
Our goal is to describe the class of graphs for which $N_\Gamma(T)$ grows in the slowest way as $T\to \infty$.

\par
 Let us note that for natural edge lengths there is a related question about choosing a graph with maximum stabilization time $N(T)$ (see \cite{RCD16}). The number $N(T)$ in this case always stops growing. So it is interesting to find the time when this happens (it is called the \textit{stabilization time}). A recent paper \cite{DW} was devoted to this problem. The class of graphs for which the stabilization time is maximum contains a long chain of edges with an odd cycle at the end (opposite the source).

\section{Main results}

In the paper \cite{ChSh} it was proved that
$$
N(T) = N_1 T^{|E|-1} + o(T^{|E|-1}), \text{ where }
N_1 = \frac{2^{2-|V|}}{(|E|-1)!} \frac{\sum_{e\in E}
t_e}{\prod_{e\in E} t_e}
$$
This implies that the graph $\Gamma$, for which $N_\Gamma(T)$ grows the slowest, must have the maximum number of vertices for a given number of edges.
Therefore, $\Gamma$ must be a tree. Accordingly, the comparison of
$N_\Gamma(T)$ for different rooted trees $\Gamma$ should take place at the level of the second term of the $N(T)$ asymptotics.

\par

Let us note that the answer to a related question about the class of graphs of maximum growth can be obtained only from the first term of the asymptotics. The maximum growth is obtained by the complete graph $K_{|E|}$ on $|E|$ edges.

\par
In the paper \cite{Chernyshev2017} it was proved that for a rooted metric tree $\Gamma$
$$
N(T) = N_1 T^{|E|-1} + N_2 T^{|E|-2} + o(T^{|E|-2}),
$$
where
$$ N_2 =
\frac{1}{(|E|-2)! 2^{|E|-2}  \prod_{e\in E} t_e}
P_2(t_1,\ldots,t_{|E|})
$$
\begin{equation}
  2 P_2(t_1,\ldots,t_{|E|}) = -\sum_{e,f\in E} t_e t_f -
  \sum_{e\in E} t_e^2 - \sum_{e\in H} t_e^2+ 2 \sum_{e\in H, f\in E}
  t_e t_f + \sum_{e\in E} \sum_{f\in up(e)} t_e t_f,
  \label{T2}
\end{equation}
where $H$ is the set of \textit{hanging} edges, that is, such that the edge has an end of degree 1 that is not a root; $up(e)$ is a simple chain leading from the end of the edge $e$ to the root of the tree (the edge itself is included in this chain). Note that the first two terms in the formula for $P_2$ do not change when the edges are permuted, while the third and fourth terms change only when the set of hanging edges changes.

\begin{figure}[!h]
\begin{center}
\includegraphics[width=4cm]{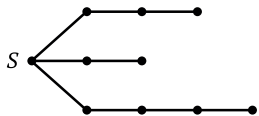}
\label{pic0}
\caption{Class of graphs, among which there is a graph of minimal growth $N(T)$}
\end{center}
\end{figure}

\begin{theorem}
  \label{th}
   Let us consider incommensurable edge lengths $t_1,\ldots,t_{|E|}$.
  If $\Gamma$ is a graph for which the function $N_\Gamma (T)$ grows minimally, then:

1) $\Gamma$ is the union of chains outgoing from the root (that is, a tree where all vertices except the root have degree 1 or 2);

2) each of these chains has at least 2 edges (see fig. \ref{pic2n}).
\end{theorem}

\begin{proof}

The main method that we use to find the class of graphs with minimal growth $N(T)$ is the transformation (surgery) of the graph $\Gamma$, which preserves the connectivity, but the asymptotic growth of $N(T)$ becomes smaller.

First transformation (T1): permutation of the edge down the so-called
\textit{branch}.

\begin{defi}
 Let us consider a non-hanging edge $e$ of the rooted tree $\Gamma$. A branch
 $br(e)$ given by an edge $e$ is a subtree consisting of edges $up(e)$ and all edges $f$ such that $e\in up(f)$ (this is the set of all descendants and ancestors of edge $e$)
 (see fig. \ref{pic1n}).
\end{defi}

\begin{figure}[!h]
\begin{center}
\includegraphics[width=4cm]{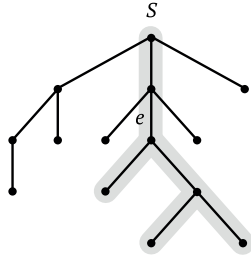}
\label{pic1n}
\caption{Example. Non-hanging edge $e$ and its branch $br(e)$.}
\end{center}
\end{figure}

Let us define the first transformation (T1) of the rooted tree. Let us take a non-hanging edge $e$ of the tree and consider its branch $br(e)$.
Consider any hanging edge $b\in br(e)$. Let us denote $l =
up(b)\setminus (\{b \}\cup up(e))$, $D = br(e) \setminus (up(e) \cup
\{b \})$.
Let us move the edge $e$ to the end of the edge
$b$ as shown in Fig. \ref{pic2n}.

\begin{figure}[!h]
\begin{center}
\includegraphics[width=16cm]{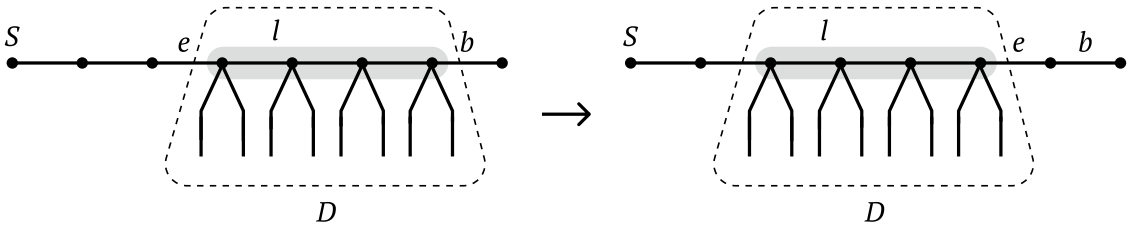}
\label{pic2n}
\caption{Transformation T1 that permutes the edge $e$ inside
$br(e)$}
\end{center}
\end{figure}

\begin{lemma}
  If the graph $\Gamma'$ is obtained from $\Gamma$ by the transformation T1, then
  $P_{2,\Gamma} \ge P_{2,\Gamma'}$. Moreover, if $D\setminus l \neq
  \emptyset$, then $P_{2,\Gamma} > P_{2,\Gamma'}$.
\end{lemma}

\begin{proof}

  During the transformation of T1, the set of hanging edges has not changed, so the first four terms in the formula \eqref{T2} have not changed. It remains to follow how the fifth term has changed.
  For edges $f$ located in $up(e)$, the set $up(f)$ has not changed.
  For the edges $f\in D$ (see Fig. \ref{pic2n}) the set $up(f)$ is decreased by the edge $e$. The set $up(b)$ has not changed. The set
 $up(e)$ has increased by the edges of a simple chain $l$ (see Fig.
 \ref{pic2n}). Therefore
 $$
 2P_{2,\Gamma} - 2 P_{2,\Gamma'} =  t_e  \sum_{f\in D\setminus l} t_f
 \ge 0
 $$
\end{proof}

Let us define a transformation (T2) that transfers a hanging edge to a neighboring edge.
Let the graph have a hanging edge $b$, one of the ends of which has a degree greater than 2. The transformation T2 joins this edge to the neighboring edge $d$
(see Fig. \ref{pic4n}).

\begin{figure}[!th]
\begin{center}
\includegraphics[width=10cm]{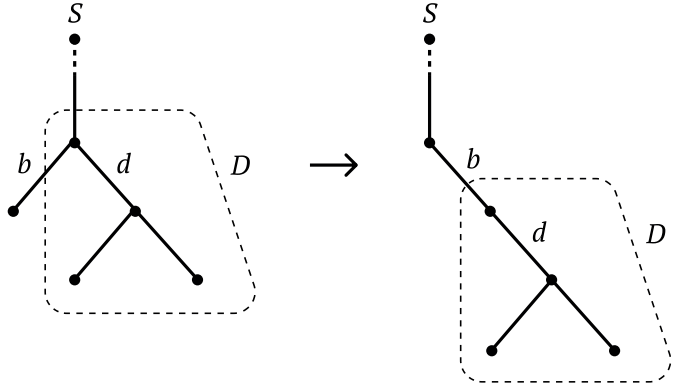}
\label{pic4n}
\caption{Transformation T2}
\end{center}
\end{figure}

\begin{lemma}
  If the graph $\Gamma'$ is obtained from $\Gamma$ by the T2 transformation, then
  $P_{2,\Gamma} > P_{2,\Gamma'}$.
\end{lemma}

\begin{proof}
$$
2 P_{2,\Gamma} - 2 P_{2,\Gamma'} = - \sum_{e\in D} t_e t_b + t_b^2 + 2
t_b \sum_{e\in E, e\neq b} t_e > 0
$$
\end{proof}

Let us proceed to the proof of the theorem. Let us prove part 1. We take a branch $br(e)$ of the non-hanging edge $e$. If the branch $br(e)$ is not a simple chain, then by transformation T1 we can reduce the asymptotics of $N(T)$. Thus, in the graph $\Gamma$, for which the growth of $N(T)$ is minimal, all branches must be simple chains. This means that the graph $\Gamma$ is the union of simple chains coming out of the root.

Let us prove part 2. If one of the simple chains has length 1, then it is a hanging edge, and by the transformation T2 this hanging edge can be moved into the neighboring chain, thus reducing the asymptotics of $N(T)$.

\end{proof}

\section{Examples}

A simple enumeration shows that for $|E|\le 5$ and any $t_1,\ldots,t_5$ the minimum growth graph
$N(T)$
is a simple chain (the root of which may be inside the chain).

Let us show that, by increasing $|E|$, the degree of the root of a graph of minimal growth can be made arbitrarily large by choosing close edge lengths.

Let us choose $d\in \mathbb N$. For $1\le k \le d$ consider the graphs
 $\Gamma_{l_1,\ldots,l_k}$ (from the theorem \ref{th}), which are the union of simple chains consisting of $k$ chains of lengths
$l_1,\ldots,l_k$ ($l_1+\ldots+l_k = |E|$).

In the polynomial $P_{2,\Gamma}(t_1,\ldots,t_{|E|})$
we select the symmetric part $S$ independent of $\Gamma$ and the asymmetric part
$T_\Gamma$ (it is enough to follow only the asymmetric part):
$$
2 P_{2,\Gamma}(t_1,\ldots,t_{|E|}) =  S(t_1,\ldots,t_{|E|}) +
T_\Gamma(t_1,\ldots,t_{|E|}),
$$
$$
S(t_1,\ldots,t_{|E|}) =
-\sum_{e,f\in E} t_e t_f -
\sum_{e\in E} t_e^2
$$
$$
T_\Gamma(t_1,\ldots,t_{|E|})= - \sum_{e\in H} t_e^2+ 2 \sum_{e\in H, f\in E}
t_e t_f + \sum_{e\in E} \sum_{f\in up(e)} t_e t_f, \
$$
Let us put equal arguments $x=t_1=\ldots=t_{|E|}$ into the polynomial $T_{\Gamma_{l_1,\ldots,l_k}}$
$$
 T_{\Gamma_{l_1,\ldots,l_k}}(x,\ldots,x) = x^2 \left[ \frac12
 (l_1^2+\ldots + l_k^2) + k (2|E|-1)+ \frac{|E|}{2} \right] \ge
 \frac{x^2 |E|^2}{2k} + O(|E|),
$$
where the inequality in the asymptotic estimate becomes an equality for the graph with $l_1 = \ldots =l_{k-1}= [\frac{|E|}{d}], l_k =|E|-l_1-\ldots
-l_{k-1}$.

In a small neighborhood of $x$ it is possible to take incommensurable $t_1, \ldots,
t_{|E|}$ and for this set the asymptotic inequalities will still hold.

Hence, with such a choice of $t_1,\ldots,t_{|E|}$, the graphs
$\Gamma_{l_1},\ldots,\Gamma_{l_1,\ldots,l_{d-1}}$ are not graphs of minimal growth $N(T)$. Therefore, a graph of minimal growth
$N(T)$
has a root degree not less than $d$.

Using computer enumeration over all partitions of the number
$|E|=24$, it can be shown that the graph of minimal growth
$N(T)$ (in the case of almost identical
$t_1$,\ldots,$t_{|E|}$) consists of exactly three chains of length 8.

\subsection{Acknowledgment}
The authors thank A.\,I. Shafarevich and L.\,V. Dworzanski for helpful discussions.

\bibliographystyle{plain}

\end{document}